\documentclass{llncs}

\usepackage{proof,latexsym, amsmath, amsfonts, amssymb, alltt, color,
algorithm, algpseudocode, cite, stmaryrd, centernot, bcprules}

\definecolor{redd}{rgb}{0.7,0.35,0}
\definecolor{bluee}{rgb}{0,0.4,0.7}
\definecolor{gree}{rgb}{0.25,0.53,0.15}

\newcommand{\ind}{\hspace{\algorithmicindent}}
\newcommand{\ret}{\textbf{return} }
\newcommand{\Match}[1]{\textbf{match} #1 \textbf{with}}
\newcommand{\Let}[2]{\textbf{let} #1 $=$ #2 \textbf{in}}

\algtext*{EndProcedure}
\algtext*{EndIf}

\newcommand{\seq}[1]{\widetilde{#1}}
\newcommand{\hccs}{\mathcal{H}}
\newcommand{\hc}{\mathit{hc}}

\newcommand{\fvs}{\mathit{fvs}}
\newcommand{\pvs}{\mathit{pvs}}
\newcommand{\dom}{\mathrm{dom}}

\newcommand{\Pvar}{\mathcal{P}}
 \newcommand{\Subst}{\Theta}
\newcommand{\subst}{\theta}

\newcommand{\SubstAtom}{\Theta_{\mathit{atom}}}

\newcommand{\maxi}{\textit{maximize}}
\newcommand{\mini}{\textit{minimize}}

\newcommand{\TINT}{\mathtt{int}}
\newcommand{\SUM}{\mathtt{sum}}
\newcommand{\SUMP}{\mathtt{sum}'}
\newcommand{\SUMT}{\mathtt{sum\_t}}
\newcommand{\TRUE}{\mathtt{true}}
\newcommand{\FALSE}{\mathtt{false}}

\newcommand{\LET}{\texttt{let}}
\newcommand{\IN}{\texttt{in}}
\newcommand{\OP}{\mathit{op}}

\newcommand{\IFZ}{\texttt{ifz}}
\newcommand{\THEN}{\texttt{then}}
\newcommand{\ELSE}{\texttt{else}}

\newcommand{\set}[1]{\left\{#1\right\}}

\newcommand{\arity}[1]{\mathit{ar}(#1)}
\newcommand{\length}[1]{\left| #1 \right|}
\newcommand{\sem}[1]{\llbracket #1 \rrbracket}
\newcommand{\Implies}{\Rightarrow}

\newcommand{\RED}{\longrightarrow_D}
\newcommand{\REDS}{\longrightarrow_D^*}

\newcommand{\rt}[2]{\left\{ #1 \mid #2 \right\}}

\newcommand\rulesp{\vspace*{2ex}}

\newcommand\MAX{\uparrow}
\newcommand\MIN{\downarrow}

\newcommand{\GEN}{\mathit{Gen}}

\newcommand{\qedex}{\qed}

\spnewtheorem*{proofsketch}{Proof Sketch}{\itshape}{\rmfamily}

\begin{document}

\title{Refinement Type Inference via Horn Constraint Optimization}
\author{Kodai Hashimoto and Hiroshi Unno}
\institute{University of Tsukuba\\
\email{\{kodai, uhiro\}@logic.cs.tsukuba.ac.jp}
}

\maketitle

\begin{abstract}
We propose a novel method for inferring refinement types of
higher-order functional programs.  The main advantage of the proposed
method is that it can infer maximally preferred (i.e., Pareto optimal)
refinement types with respect to a user-specified preference order.
The flexible optimization of refinement types enabled by the proposed
method paves the way for interesting applications, such as inferring
most-general characterization of inputs for which a given program
satisfies (or violates) a given safety (or termination) property.  Our
method reduces such a type optimization problem to a Horn constraint
optimization problem by using a new refinement type system that can
flexibly reason about non-determinism in programs.  Our method then
solves the constraint optimization problem by repeatedly improving a
current solution until convergence via template-based invariant
generation.  We have implemented a prototype inference system based on
our method, and obtained promising results in preliminary experiments.
\end{abstract}

\section{Introduction}
\label{sec:intro}

Refinement types~\cite{Freeman1991,Xi1999} have been applied to safety
verification of higher-order functional programs.  Some existing
tools~\cite{Unno2008,Rondon2008,Unno2009,Kobayashi2011b,Jhala2011,Terauchi2010,Unno2013}
enable fully automated verification by refinement type inference based
on invariant generation techniques such as abstract interpretation,
predicate abstraction, and CEGAR.  The goal of these tools is to infer
refinement types precise enough to verify a given safety
specification.  Therefore, types inferred by these tools are often too
specific to the particular specification, and hence have limited
applications.

To remedy the limitation, we propose a novel refinement type inference
method that can infer maximally preferred (i.e., Pareto optimal)
refinement types with respect to a user-specified preference order.
For example, let us consider the following summation function (in
OCaml syntax)
\begin{alltt}
  let rec sum x = if x = 0 then 0 else x + sum (x - 1)
\end{alltt}
A refinement type of $\SUM$ is of the form $(x:\rt{x:\TINT}{P(x)}) \to
\rt{y:\TINT}{Q(x,y)}$.  Here, $P(x)$ and $Q(x,y)$ respectively
represent pre and post conditions of $\SUM$.  Note that the
postcondition $Q(x,y)$ can refer to the argument $x$ as well as the
return value $y$.  Suppose that we want to infer a maximally-weak
predicate for $P$ and maximally-strong predicate for $Q$ within a
given underlying theory.  Our method allows us to specify such
preferences as the following constraints for type optimization
\[
\maxi(P), \qquad
\mini(Q).
\]
Here, $\maxi(P)$ (resp. $\mini(Q))$ means that the set of the models
of $P(x)$ (resp. $Q(x,y)$) should be maximized (resp. minimized).  Our
method then infers a Pareto optimal refinement type with respect to
the given preferences.

In general, however, this kind of multi-objective optimization
involves a trade-off among the optimization constraints.  In the above
example, $P$ may not be weakened without also weakening $Q$.  Hence,
there often exist multiple optima.  Actually, all the following are
Pareto optimal refinement types of $\SUM$.\footnote{Here, we use
  quantifier-free linear arithmetic as the underlying theory and
  consider only atomic predicates for $P$ and $Q$.}
\begin{eqnarray}
(x:\rt{x:\TINT}{x=0}) &\to& \rt{y:\TINT}{y=x} \\
(x:\rt{x:\TINT}{\TRUE}) &\to& \rt{y:\TINT}{y \geq 0} \\
(x:\rt{x:\TINT}{x < 0}) &\to& \rt{y:\TINT}{\FALSE}
\end{eqnarray}

Our method further allows us to specify a priority order on the
predicate variables $P$ and $Q$.  If $P$ is given a higher priority
over $Q$ (we write $P \sqsubset Q$), our method infers the type (2),
whereas we obtain the type (3) if $Q \sqsubset P$.
Interestingly, (3) expresses that $\SUM$ is non-terminating for any
input $x < 0$.

The flexible optimization of refinement types enabled by our method
paves the way for interesting applications, such as inferring
most-general characterization of inputs for which a given program
satisfies (or violates) a given safety (or termination) property.
Furthermore, our method can infer an upper bound of the number of
recursive calls if the program is terminating, and can find a
minimal-length counterexample path if the program violates a safety
property.

Internally, our method reduces such a refinement type optimization
problem to a constraint optimization problem where the constraints are
expressed as existentially quantified Horn clauses over predicate
variables~\cite{Unno2013,Beyene2013,Kuwahara2015}.
The constraint generation here is based on a new refinement type
system that can reason about (angelic and demonic) non-determinism in
programs.  Our method then solves the constraint optimization problem
by repeatedly improving a current solution until convergence.  The
constraint optimization here is based on an extension of
template-based invariant generation~\cite{Colon2003,Gulwani2008a} to
existentially quantified Horn clause constraints and prioritized
multi-objective optimization.

The rest of the paper is organized as follows.
Sections~\ref{sec:lang} and \ref{sec:type} respectively formalize our
target language and its refinement type system.  The applications of
refinement type optimization are explained in Section~\ref{sec:app}.
Section~\ref{sec:opt} formalizes Horn constraint optimization problems
and the reduction from type optimization problems.
Section~\ref{sec:alg} proposes our Horn constraint optimization
method.  Section~\ref{sec:exp} reports on a prototype implementation
of our method and the results of preliminary experiments.  We compare
our method with related work in Section~\ref{sec:rel} and conclude the
paper in Section~\ref{sec:concl}.

\section{Target Language $L$}
\label{sec:lang}

This section introduces a higher-order call-by-value functional
language $L$, which is the target of our refinement type optimization.
The syntax is defined as follows.
\[
\begin{array}{rrl}
\text{(programs)}
 &D ::=& \set{f_i\ \seq{x}_i = e_i}_{i=1}^m \\
\text{(expressions)}
 &e ::=& x \mid e\ v \mid n \mid \OP(v_1,\dots,v_{\arity{\OP}}) \mid
         \IFZ\ v\ \THEN\ e_1\ \ELSE\ e_2 \\
&\mid\ & \LET\ x= e_1\ \IN\ e_2 \mid \LET\ x= *_{\forall}\ \IN\ e \mid \LET\ x= *_{\exists}\ \IN\ e \\
\text{(values)}
 &v ::=& x \mid x\ \seq{v} \mid n \\
\text{(eval. contexts)}
 &E ::=& [\;] \mid E\ v \mid \LET\ x\ =\ E\ \IN\ e
\end{array}
\]
Here, $x$ and $f$ are meta-variables ranging over variables.  $n$ and
$\OP$ respectively represent integer constants and operations such as
$+$ and $\geq$.  $\arity{\OP}$ expresses the arity of $\OP$.  We write
$\seq{x}$ (resp. $\seq{v}$) for a sequence of variables
(resp. values).  For simplicity of the presentation, the language $L$
has integers as the only data type.  We encode Boolean values $\TRUE$
and $\FALSE$ respectively as non-zero integers and $0$.  A program is
expressed as a set $\set{f_i\ \seq{x}_i = e_i}_{i=1}^m$ of function
definitions.  We define $\dom(\set{f_i\ \seq{x}_i =
  e_i}_{i=1}^m)=\set{f_1,\dots,f_m}$.
We assume that a value $x\ \seq{v}$ satisfies $1 \leq \length{\seq{v}}
< \arity{f}$, where $\length{\seq{v}}$ represents the length of the
sequence $\seq{v}$.

The call-by-value operational semantics of $L$ is given in
Figure~\ref{fig:semantics}.  Here, $\sem{\OP}$ represents the integer
function denoted by $\OP$.  Both expressions $\LET\ x=
*_{\forall}\ \IN\ e$ and $\LET\ x= *_{\exists}\ \IN\ e$ generate a
random integer $n$, bind $x$ to it, and evaluate $e$.  They are,
however, interpreted differently in our refinement type system (see
Section~\ref{sec:type}).  These expressions are used to model external
functions without definitions and non-deterministic behavior caused by
external inputs (e.g., user inputs, interrupts, and so on).  We write
$\REDS$ to denote the reflexive and transitive closure of $\RED$.

\begin{figure}[t]
\begin{minipage}[t]{0.5\hsize}

\infax[E-\textsc{Op}]{E[\OP(\seq{v})] \RED E[\sem{\OP}(\seq{v})]}
\rulesp

\infrule[E-\textsc{App}]
  {f\ \seq{x}=e \in D \andalso \length{\seq{x}} = \length{\seq{v}}}
  {E[f\ \seq{v}] \RED E[[\seq{v}/\seq{x}]e]}
\rulesp

\infax[E-\textsc{Let}]{E[\LET\ x= v\ \IN\ e] \RED E[[v/x]e]}

\end{minipage}
\begin{minipage}[t]{0.5\hsize}

\infax[E-\textsc{Rand}$\exists$]{E[\LET\ x= *_{\forall}\ \IN\ e] \RED E[[n/x]e]}
\rulesp

\infax[E-\textsc{Rand}$\forall$]{E[\LET\ x= *_{\exists}\ \IN\ e] \RED E[[n/x]e]}
\rulesp

\infrule[E-\textsc{If}]
  {\mbox{ if } n=0 \mbox{ then } i=1 \mbox{ else } i=2}
  {E[\IFZ\ n\ \THEN\ e_1\ \ELSE\ e_2] \RED E[e_i]}

\end{minipage}
\caption{The operational semantics of our language $L$}
\label{fig:semantics}
\end{figure}

\section{Refinement Type System for $L$}
\label{sec:type}

In this section, we introduce a refinement type system for $L$ that
can reason about non-determinism in programs.  We then formalize
refinement type optimization problems (in Section~\ref{sec:type_opt}),
which generalize ordinary type inference problems.

The syntax of our refinement type system is defined as follows.
\[
\begin{array}{rrl}
\text{(refinement types)}
 &\tau ::=& \rt{x}{\phi} \mid (x:\tau_1)\to \tau_2 \\
\text{(type environments)}
 &\Gamma ::=& \emptyset \mid \Gamma,x:\tau \mid \Gamma,\phi \\
\text{(formulas)}
 &\phi ::=& t_1 \leq t_2 \mid \top \mid \bot \mid \lnot \phi \mid \phi_1 \land \phi_2 \mid \phi_1 \lor \phi_2 \mid \phi_1 \Implies \phi_2 \\
\text{(terms)}
 &t ::=& n \mid x \mid t_1+t_2 \mid n \cdot t \\
\text{(predicates)}
 &p ::=& \lambda \seq{x}.\phi
\end{array}
\]
An integer refinement type $\rt{x}{\phi}$ equipped with a formula
$\phi$ for type refinement represents the type of integers $x$ that
satisfy $\phi$.  The scope of $x$ is within $\phi$.  We often
abbreviate $\rt{x}{\top}$ as $\TINT$.  A function refinement type
$(x:\tau_1)\to\tau_2$ represents the type of functions that take an
argument $x$ of the type $\tau_1$ and return a value of the type
$\tau_2$.  Here, $\tau_2$ may depend on the argument $x$ and the scope
of $x$ is within $\tau_2$.  For example, $(x:\TINT) \to \rt{y}{y > x}$
is the type of functions whose return value $y$ is always greater than
the argument $x$.  We often write $\fvs(\tau)$ to denote the set of
free variables occurring in $\tau$.  We define $\Gamma(x)=\tau$ if
$x:\tau \in \Gamma$ and $\dom(\Gamma)=\set{x \mid x:\tau \in \Gamma}$.

In this paper, we adopt formulas $\phi$ of the quantifier-free theory
of linear integer arithmetic (QFLIA) for type refinement.  We write
$\models \phi$ if a formula $\phi$ is valid in QFLIA.  Formulas $\top$
and $\bot$ respectively represent the tautology and the contradiction.
Note that atomic formulas $t_1 < t_2$ (resp. $t_1 = t_2$) can be
encoded as $t_1 + 1 \leq t_2$ (resp. $t_1 \leq t_2 \land t_2 \leq
t_1$) in QFLIA.

\begin{figure}[t]
\typicallabel{Rand$\forall$}
\begin{minipage}[t]{0.4\hsize}
\infrule[Prog]
  {\Gamma \vdash \lambda \seq{x}_i.e_i : \Gamma(f_i) \\
   (\mbox{for }i \in \set{1,\dots,m})}
  {\vdash \set{f_i\ \seq{x}_i = e_i}_{i=1}^m : \Gamma}
\rulesp

\infrule[IVar]
  {\Gamma(x)=\rt{\nu}{\phi}}
  {\Gamma \vdash x : \rt{\nu}{\nu=x}}
\rulesp

\infrule[FVar]
  {\Gamma(x)=(\nu:\tau_1) \to \tau_2}
  {\Gamma \vdash x : (\nu:\tau_1) \to \tau_2}
\rulesp

\infrule[Abs]
  {\Gamma,x:\tau_1 \vdash e : \tau_2}
  {\Gamma \vdash \lambda x. e : (x:\tau_1) \to \tau_2}
\rulesp

\infrule[App]
  {\Gamma \vdash e : (x:\tau_1) \to \tau_2 \\
   \Gamma \vdash v : \tau_1}
  {\Gamma \vdash e\ v : [v/x]\tau_2}
\rulesp

\infax[Int]
  {\Gamma \vdash n : \rt{\nu}{\nu=n}}
\rulesp

\infrule[Sub]
  {\Gamma \vdash e : \tau' \andalso
   \Gamma \vdash \tau' <: \tau}
  {\Gamma \vdash e : \tau}
\rulesp

\infrule[ISub]
  {\models \sem{\Gamma} \land \phi_1 \Implies \phi_2}
  {\Gamma \vdash \rt{\nu}{\phi_1} <: \rt{\nu}{\phi_2}}

\end{minipage}
\begin{minipage}[t]{0.6\hsize}

\infrule[Let]
  {\Gamma \vdash e_1 : \tau_1 \\
   \Gamma,x : \tau_1 \vdash e_2 : \tau_2 \andalso
   x \not\in \fvs(\tau_2)}
  {\Gamma \vdash \LET\ x= e_1\ \IN\ e_2 : \tau_2}
\rulesp

\infrule[Rand$\forall$]
  {\Gamma,x : \TINT \vdash e : \tau \andalso
   x \not\in \fvs(\tau)}
  {\Gamma \vdash \LET\ x= *_{\forall}\ \IN\ e : \tau}
\rulesp

\infrule[Rand$\exists$]
  {\fvs(\phi) \subseteq \dom(\Gamma) \cup \set{x} \\
   \models \sem{\Gamma} \Implies \exists x. \phi \\
   \Gamma,x : \rt{x}{\phi} \vdash e : \tau \andalso
   x \not\in \fvs(\tau)}
  {\Gamma \vdash \LET\ x= *_{\exists}\ \IN\ e : \tau}
\rulesp

\infrule[If]
  {\Gamma,v=0 \vdash e_1 : \tau \\
   \Gamma,v \neq 0 \vdash e_2 : \tau}
  {\Gamma \vdash \IFZ\ v\ \THEN\ e_1\ \ELSE\ e_2 : \tau}
\rulesp

\infrule[Op]
  {\sem{\OP}^{\mathrm{Ty}}=(x_1:\tau_1) \to \dots \to (x_m:\tau_m) \to \tau \\
   \sigma_j =[v_1/x_1,\dots,v_j/x_j] \\
   \Gamma \vdash v_i : \sigma_{i-1} \tau_i \quad (\mbox{for }i \in \set{1,\dots,m})}
  {\Gamma \vdash \OP(v_1,\dots,v_m) : \sigma_m \tau}
\rulesp

\infrule[FSub]
  {\Gamma \vdash \tau_1' <: \tau_1 \\
   \Gamma,\nu:\tau_1' \vdash \tau_2 <: \tau_2'}
  {\Gamma \vdash (\nu:\tau_1) \to \tau_2 <: (\nu:\tau_1') \to \tau_2'}

\end{minipage}

\caption{The inference rules of our refinement type system}
\label{fig:type}
\end{figure}

The inference rules of our refinement type system are shown in
Figure~\ref{fig:type}.  Here, a type judgment $\vdash D : \Gamma$
means that a program $D$ is well-typed under a refinement type
environment $\Gamma$.  A type judgment $\Gamma \vdash e : \tau$
indicates that an expression $e$ has a refinement type $\tau$ under
$\Gamma$.  A subtype judgment $\Gamma \vdash \tau_1 <: \tau_2$ states
that $\tau_1$ is a subtype of $\tau_2$ under $\Gamma$.  $\sem{\Gamma}$
occurring in the rules \rn{ISub} and \rn{Rand$\exists$} is defined by
$\sem{\emptyset}=\top$, $\sem{\Gamma,x:\rt{\nu}{\phi}}=\sem{\Gamma}
\land [x/\nu]\phi$, $\sem{\Gamma,x:(\nu:\tau_1) \to
  \tau_2}=\sem{\Gamma}$, and $\sem{\Gamma,\phi}=\sem{\Gamma} \land
\phi$.  In the rule \rn{Op}, $\sem{\OP}^{\mathrm{Ty}}$ represents a
refinement type of $\OP$ that soundly abstracts the behavior of the
function $\sem{\OP}$.  For example, $\sem{+}^{\mathrm{Ty}}=(x:\TINT)
\to (y:\TINT) \to \rt{z}{z=x+y}$.
All the rules except \rn{Rand$\forall$} and \rn{Rand$\exists$} for
random integer generation are essentially the same as the previous
ones~\cite{Unno2009}.  The rule \rn{Rand$\forall$} requires $e$ to
have $\tau$ for \emph{any} randomly generated integer $x$.  Therefore,
$e$ is type-checked against $\tau$ under a type environment that
assigns $\TINT$ to $x$.  By contrast, the rule \rn{Rand$\exists$}
requires $e$ to have $\tau$ for \emph{some} randomly generated integer
$x$.  Hence, $e$ is type-checked against $\tau$ under a type
environment that assigns a type $\rt{x}{\phi}$ to $x$ for some $\phi$
such that $\fvs(\phi) \subseteq \dom(\Gamma) \cup \set{x}$ and
$\models \sem{\Gamma} \Implies \exists x. \phi$.
For example, $x:\TINT \vdash \LET\ y= *_{\exists}\ \IN\ x+y :
\rt{r}{r=0}$ is derivable because we can derive
$x:\TINT,y:\rt{y}{y=-x} \vdash x+y:\rt{r}{r=0}$.
Thus, our new type system allows us to reason about both angelic
$*_{\exists}$ and demonic $*_{\forall}$ non-determinism in
higher-order functional programs.

We now discuss properties of our new refinement type system.  We can
prove the following progress theorem in a standard manner.
\begin{theorem}[Progress]
\label{thm:progress}
Suppose that we have $\vdash D : \Gamma$, $\dom(\Gamma)=\dom(D)$, and
$\Gamma \vdash e : \tau$.  Then, either $e$ is a value or $e \RED e'$
for some $e'$.
\end{theorem}
We can also show the substitution lemma and the type preservation
theorem in a similar manner to \cite{Unno2009}.
\begin{lemma}[Substitution]
\label{lem:subst}
If $\Gamma\vdash v:\tau'$ and $\Gamma,x : \tau',\Gamma'\vdash e : \tau$,
then $\Gamma,[v/x]\Gamma'\vdash[v/x]e:[v/x]\tau$.
\end{lemma}
\begin{theorem}[Preservation]
\label{thm:preserve}
Suppose that we have $\vdash D : \Gamma$ and $\Gamma \vdash e : \tau$.
If $e$ is of the form \emph{$\LET\ x= *_{\exists}\ \IN\ e_0$}, then we
get $\Gamma \vdash e' : \tau$ for \emph{some} $e'$ such that $e \RED
e'$.  Otherwise, we get $\Gamma \vdash e' : \tau$ for \emph{any} $e'$
such that $e \RED e'$.
\end{theorem}
\begin{proof}
We prove the theorem by induction on the derivation of $\Gamma \vdash
e : \tau$.  We only show the case for the rule \rn{Rand$\exists$}
below.  The other cases are similar to \cite{Unno2009}.  By
\rn{Rand$\exists$}, we have $e = \LET\ x= *_{\exists}\ \IN\ e_0$,
$\fvs(\phi) \subseteq \dom(\Gamma) \cup \set{x}$, $\models
\sem{\Gamma} \Implies \exists x. \phi$, $\Gamma,x : \rt{x}{\phi}
\vdash e_0 : \tau$, and $x \not\in \fvs(\tau)$.  It then follows from
$\models \sem{\Gamma} \Implies \exists x. \phi$ that there is an
integer $n$ such that $\models \sem{\Gamma} \land x=n \Implies \phi$.
By the rule $\textsc{E-Rand}\exists$, we get $e \RED [n/x]e_0=e'$.  By
the rules \textsc{Int} and $\textsc{Sub}$, we obtain $\Gamma\vdash
n:\rt{x}{\phi}$.  Thus, we get $\Gamma \vdash e' : \tau$ by
Lemma~\ref{lem:subst}, $\Gamma,x : \rt{x}{\phi} \vdash e_0 : \tau$,
and $x \not\in \fvs(\tau)$.  \qed
\end{proof}

\subsection{Refinement Type Optimization Problems}
\label{sec:type_opt}

We now define refinement type optimization problems, which generalize
refinement type inference problems addressed by previous
work~\cite{Unno2008,Rondon2008,Unno2009,Kobayashi2011b,Jhala2011,Terauchi2010,Unno2013}.

We first introduce the notion of \emph{refinement type templates}.  A
refinement type template of a function $f$ is the refinement type
obtained from the ordinary ML-style type of $f$ by replacing each base
type $\TINT$ with an integer refinement type
$\rt{\nu}{P(\seq{x},\nu)}$ for some fresh predicate variable $P$ that
represents an unknown predicate to be inferred, and each function type
$T_1 \to T_2$ with a (dependent) function refinement type $(x:\tau_1)
\to \tau_2$.  For example, from an ML-style type $(\TINT \to \TINT)
\to \TINT \to \TINT$, we obtain the following template.
\begin{align*}
&\left(f:\ \left(x_1:\rt{x_1}{P_1(x_1)}\right) \to \rt{x_2}{P_2(x_1,x_2)}\right) \to \\
&\ (x_3:\rt{x_3}{P_3(x_3)}) \to \rt{x_4}{P_4(x_3,x_4)}
\end{align*}
Note here that the first argument $f$ is not passed as an argument to
$P_3$ and $P_4$ because $f$ is of a function type and never occurs in
QFLIA formulas for type refinement.  A refinement type template of a
program $D$ with $\dom(D)=\set{f_1,\dots,f_m}$ is the refinement type
environment $\Gamma_D=f_1:\tau_1,\dots,f_m:\tau_m$, where each
$\tau_i$ is the refinement type template of $f_i$.  We write
$\pvs(\Gamma_D)$ for the set of predicate variables that occur in
$\Gamma_D$.
A \emph{predicate substitution} $\subst$ for $\Gamma_D$ is a map from
each $P \in \pvs(\Gamma_D)$ to a closed predicate $\lambda
x_1,\dots,x_{\arity{P}}.\phi$, where $\arity{P}$ represents the arity
of $P$.  We write $\subst \Gamma_D$ to denote the application of a
substitution $\subst$ to $\Gamma_D$.  We also write $\dom(\subst)$ to
represent the domain of $\subst$.

We can define ordinary refinement type inference problems as follows.
\begin{definition}[Refinement Type Inference]
\label{def:type_inf}
A \emph{refinement type inference problem} of a program $D$ is a
problem to find a predicate substitution $\subst$ such that $\vdash D
: \subst \Gamma_D$.
\end{definition}
We now generalize refinement type inference problems to optimization
problems.
\begin{definition}[Refinement Type Optimization]
\label{def:type_opt}
Let $D$ be a program, $\prec$ be a strict partial order on predicate
substitutions, and $\Subst=\set{\subst \mid \ \vdash D : \subst
  \Gamma_D}$.  A predicate substitution $\subst \in \Subst$ is called
\emph{Pareto optimal with respect to $\prec$} if there is no $\subst'
\in \Subst$ such that $\subst' \prec \subst$.  A \emph{refinement type
  optimization problem} $(D,\prec)$ is a problem to find a Pareto
optimal substitution $\subst \in \Subst$ with respect to $\prec$.
\end{definition}
In the remainder of the paper, we will often consider type
optimization problems extended with user-specified constraints and/or
templates for some predicate variables (see Section~\ref{sec:app} for
examples and Section~\ref{sec:opt} for formal definitions).

The above definition of type optimization problems is abstract in the
sense that $\prec$ is only required to be a strict partial order on
predicate substitutions.  We below introduce an example concrete
order, which is already explained informally in
Section~\ref{sec:intro} and adopted in our prototype implementation
described in Section~\ref{sec:exp}.  The order is defined by two kinds
of optimization constraints: the optimization direction
(i.e. minimize/maximize) and the priority order on predicate variables.

\begin{definition}
Suppose that
\begin{itemize}
\item $\Pvar=\set{P_1,\dots,P_m}$ is a subset of $\pvs(\Gamma_D)$,
\item $\rho$ is a map from each predicate variable in $\Pvar$ to an
  optimization direction $d$ that is either $\MAX$ (for maximization)
  or $\MIN$ (for minimization), and
\item $\sqsubset$ is a strict total order on $\Pvar$ that expresses
  the priority.\footnote{If $\sqsubset$ were partial, the relation
    $\prec_{(\rho,\sqsubset)}$ defined shortly would not be a strict
    partial order.  Our implementation described in
    Section~\ref{sec:exp} uses topological sort to obtain a strict
    total order $\sqsubset$ from a user-specified partial one.}  We
  below assume that $P_1 \sqsubset \dots \sqsubset P_m$.
\end{itemize}
We define a strict partial order $\prec_{(\rho,\sqsubset)}$ on
predicate substitutions that respects $\rho$ and $\sqsubset$ as the
following lexicographic order:
\[
\subst_1 \prec_{(\rho,\sqsubset)} \subst_2
\Longleftrightarrow \exists i \in \set{1,\dots,m}.\ \theta_1(P_i) \prec_{\rho(P_i)} \theta_2(P_i) \land
                    \forall j < i.\ \theta_1(P_j) \equiv_{\rho(P_j)} \theta_2(P_j)
\]
Here, a strict partial order $\prec_d$ and an equivalence relation
$\equiv_d$ on predicates are defined as follows.
\begin{itemize}
\item $p_1 \prec_d p_2 \Longleftrightarrow p_1 \preceq_d p_2 \land p_2 \not\preceq_d p_1$,
\item $p_1 \equiv_d p_2 \Longleftrightarrow p_1 \preceq_d p_2 \land p_2 \preceq_d p_1$,
\item $\lambda\seq{x}.\phi_1 \preceq_{\MAX} \lambda\seq{x}.\phi_2 \Longleftrightarrow \models \phi_2 \Implies \phi_1$, and
      $\lambda\seq{x}.\phi_1 \preceq_{\MIN} \lambda\seq{x}.\phi_2 \Longleftrightarrow \models \phi_1 \Implies \phi_2$.
\end{itemize}
\end{definition}

\begin{example}
\label{ex:sum}
Recall the function $\SUM$ and its type template with the predicate
variables $P,Q$ in Section~\ref{sec:intro}.  Let us consider
optimization constraints $\rho(P) = \MAX$, $\rho(Q) = \MIN$, and
$P\sqsubset Q$, and predicate substitutions
\begin{itemize}
\item $\subst_1 = \set{P \mapsto \lambda x.\, x=0,\, Q \mapsto \lambda
    x,y.\, y=x}$,
\item $\subst_2 = \set{P \mapsto \lambda x.\, \top,\, Q \mapsto
  \lambda x,y.\, y \geq 0}$, and
\item $\subst_3 = \{P \mapsto \lambda x.\, x<0,\, Q \mapsto \lambda
  x,y.\, \bot\}$.
\end{itemize}
We then have $\subst_2 \prec_{(\rho,\sqsubset)} \subst_1$ and
$\subst_2 \prec_{(\rho,\sqsubset)} \subst_3$, because $(\lambda
x.\,\top) \prec_{\MAX} (\lambda x.\,x=0)$ and $(\lambda x.\,\top)
\prec_{\MAX} (\lambda x.\,x<0)$.
\qedex
\end{example}

\section{Applications of Refinement Type Optimization}
\label{sec:app}

In this section, we present applications of refinement type
optimization to the problems of proving safety (in
Section~\ref{sec:safe}) and termination (in Section~\ref{sec:term}),
and disproving safety (in Section~\ref{sec:non-safe}) and termination
(in Section~\ref{sec:non-term}) of programs in the language $L$.  In
particular, we discuss precondition inference, namely, inference of
most-general characterization of inputs for which a given program
satisfies (or violates) a given safety (or termination) property.

\subsection{Proving Safety}
\label{sec:safe}

We explain how to formalize, as a type optimization problem, a problem
of inferring maximally-weak precondition under which a given program
satisfies a given postcondition.  For example, let us consider the
following terminating version of $\SUM$.
\begin{alltt}
  let rec sum' x = if x <= 0 then 0 else x + sum' (x-1)
\end{alltt}
In our framework, a problem to infer a maximally-weak precondition on
the argument $x$ for a postcondition $x = \SUMP\ x$ is expressed as a
type optimization problem to infer $\SUMP$'s refinement type of the
form $(x: \rt{x}{P(x)}) \to \rt{y}{x = y}$ under an optimization
constraint $\maxi(P)$.
Our type optimization method described in Sections~\ref{sec:red} and
\ref{sec:alg} infers the following type.
\[
(x: \rt{x}{0 \leq x \leq 1}) \to \rt{y}{x = y}
\]
This type says that the postcondition holds if the actual argument $x$
is 0 or 1.

\begin{example}[Hihger-Order Function]
For an example of a higher-order function, consider the following.
\begin{alltt}
  let rec repeat f n e = if n<=0 then e else repeat f (n-1) (f e)
\end{alltt}
By inferring \texttt{repeat}'s refinement type of the form
\[
(f:(x: \rt{x}{P_1(x)}) \to \rt{y}{P_2(x, y)}) \to (n:\TINT) \to (e:\rt{e}{P_3(n, e)}) \to \rt{r}{r \geq 0}
\]
under optimization constraints $\rho(P_1)=\MIN$,
$\rho(P_2)=\rho(P_3)=\MAX$, and $P_3 \sqsubset P_2 \sqsubset P_1$, our
type optimization method obtains
\[
(f:(x: \rt{x}{x \geq 0}) \to \rt{y}{y \geq 0}) \to (n:\TINT) \to (e:\rt{e}{e \geq 0}) \to \rt{r}{r \geq 0}
\]
Thus, type optimization can be applied to infer maximally-weak
refinement types of (possibly higher-order) arguments that are
sufficient for the function to satisfy a given postcondition.  \qedex
\end{example}

\subsection{Disproving Termination}
\label{sec:non-term}

In a similar manner to Section~\ref{sec:safe}, we can apply type
optimization to the problems of inferring maximally-weak precondition
for a given program to violate the termination property.  For example,
consider the function $\SUM$ in Section~\ref{sec:intro}.  For
disproving termination of $\SUM$, we infer $\SUM$'s refinement type of
the form $(x:\rt{x}{P(x)}) \to \rt{y}{\bot}$ under an optimization
constraint $\maxi(P)$.
Our type optimization method infers the following type.
\[
(x: \rt{x}{x < 0}) \to \rt{y}{\bot}
\]
The type expresses the fact that no value is returned by $\SUM$ (i.e.,
$\SUM$ is non-terminating) if the actual argument $x$ satisfies $x <
0$.

\begin{example}[Non-Deterministic Function]
For an example of non-deterministic function, let us consider a
problem of disproving termination of the following.
\begin{alltt}
  let rec f x = let n = read_int () in if n<0 then x else f x
\end{alltt}
Here, $\verb|read_int ()|$ is a function to get an integer value from
the user and is modeled as $*_{\exists}$ in our language $L$.  Note
that the termination of \texttt{f} does not depend on the argument
\texttt{x} but user inputs \texttt{n}.  Actually, our type
optimization method successfully disproves termination of \texttt{f}
by inferring a refinement type $(x:\TINT) \to \rt{y}{\bot}$ for
\texttt{f} and $\rt{n}{n\geq0}$ for the user inputs \texttt{n}.  This
means that \texttt{f} is non-terminating if the user always inputs
some non-negative integer.  \qedex
\end{example}

\subsection{Proving Termination}
\label{sec:term}

Refinement type optimization can also be applied to bounds analysis
for inferring upper bounds of the number of recursive calls.
Our bounds analysis for functional programs is inspired by a program
transformation approach to bounds analysis for imperative
programs~\cite{Gulwani2008a,Gulwani2009a}.  Let us consider $\SUM$ in
Section~\ref{sec:intro}.  By inserting additional parameters $i$ and
$c$ to the definition of $\SUM$, we obtain
\begin{alltt}
  let rec sum_t x i c = if x=0 then 0 else x + sum_t (x-1) i (c+1)
\end{alltt}
Here, $i$ and $c$ respectively represent the initial value of the
argument $x$ and the number of recursive calls so far.  For proving
termination of $\SUM$, we infer $\SUMT$'s refinement type of the form
\[
(x:\rt{x}{P(x}) \to (i:\TINT) \to (c:\rt{c}{\mathit{Inv}(x,i,c)}) \to \TINT
\]
under optimization constraints $\maxi(P)$, $\mini(\mathit{Bnd})$, $P
\sqsubset \mathit{Bnd}$, and additional constraints on the predicate
variables $P,\mathit{Bnd},\mathit{Inv}$
\begin{align}
\forall x,i,c.\ (\mathit{Inv}(x,i,c)\Leftarrow&\ c = 0 \land i = x) \\
\forall x,i,c.\ (\mathit{Bnd}(i,c) \Leftarrow&\ P(x) \land Inv(x,i,c))
\end{align}
Here, $\mathit{Bnd}(i,c)$ is intended to represent the bounds of the
number $c$ of recursive calls of $\SUM$ with respect to the initial
value $i$ of the argument $x$.  We therefore assume that
$\mathit{Bnd}(i,c)$ is of the form $0 \leq c \leq k_0 + k_1 \cdot i$,
where $k_0,k_1$ represent unknown coefficients to be inferred.
The constraint (4) is necessary to express the meaning of the inserted
parameters $i$ and $c$.  The constraint (5) is also necessary to
ensure that the bounds $\mathit{Bnd}(i,c)$ is implied by a
precondition $P(x)$ and an invariant $Inv(x,i,c)$ of $\SUM$.
Our type optimization method then infers
\[
(x:\rt{x}{x\geq 0}) \to (i:\TINT) \to (c:\rt{c}{x \leq i \land i=x+c}) \to \TINT
\]
and $\mathit{Bnd}(i,c) \equiv 0 \leq c \leq i$.  Thus, we can conclude
that $\SUM$ is terminating for any input $x\geq 0$ because the number
$c$ of recursive calls is bounded from above by the initial value $i$
of the argument $x$.

Interestingly, we can infer a precondition for minimizing the number
of recursive calls of $\SUM$ by replacing the priority constraint $P
\sqsubset \mathit{Bnd}$ with $\mathit{Bnd} \sqsubset P$ and adding an
additional constraint $\exists x.P(x)$ (to avoid a meaningless
solution $P(x) \equiv \bot$).  In fact, our type optimization method
obtains
\[
(x:\rt{x}{x=0}) \to (i:\TINT) \to (c:\rt{c}{c=0}) \to \TINT
\]
and $\mathit{Bnd}(i,c) \equiv c=0$.  Therefore, we can conclude that
the minimum number of recursive calls is $0$ when the actual argument
$x$ is $0$.

We expect that our bounds analysis for functional programs can further
be extended to infer non-linear upper bounds by adopting ideas from an
elaborate transformation for bounds analysis of imperative
programs~\cite{Gulwani2009a}.

\subsection{Disproving Safety}
\label{sec:non-safe}

We can use the same technique in Section~\ref{sec:term} to infer
maximally-weak precondition for a given program to violate a given
postcondition.  For example, let us consider again the function
$\SUM$.  A problem to infer a maximally-weak precondition on the
argument $x$ for violating a postcondition $\SUM\ x \geq 2$ can be
reduced to a problem to infer $\SUMT$'s refinement type of the form
\begin{align*}
(x:\rt{x}{P(x)}) \to (i:\TINT) \to (c:\rt{c}{\mathit{Inv}(x,i,c)}) \to \rt{y}{\neg(y \geq 2)}
\end{align*}
under the same constraints for bounds analysis in
Section~\ref{sec:term}.
The refinement type optimization method then obtains
\begin{align*}
(x:\rt{x}{0\leq x\leq 1}) \to (i:\TINT) \to (c:\rt{c}{0\leq x \land i=x+c}) \to \rt{y}{\neg(y \geq 2)}
\end{align*}
and $\mathit{Bnd}(i,c) \equiv 0 \leq c \leq i$.  This result says that
if the actual argument $x$ is 0 or 1, then $\SUM$ terminates and
returns some integer $y$ that violates $y \geq 2$.  In other words,
$x=0,1$ are counterexamples to the postcondition $\SUM\ x \geq 2$.

We can instead find a minimal-length counterexample
path\footnote{Here, minimality is with respect to the number of
  recursive calls within the path.} to the postcondition $\SUM\ x \geq
2$ by just replacing the priority constraint $P \sqsubset
\mathit{Bnd}$ with $\mathit{Bnd} \sqsubset P$ and adding an additional
constraint $\exists x.P(x)$.  Our type optimization method then infers
\[
(x:\rt{x}{x=0}) \to (i:\TINT) \to (c:\rt{c}{0\leq x \land i=x+c}) \to \rt{y}{\neg(y \geq 2)}
\]
and $\mathit{Bnd}(i,c) \equiv c=0$.  From the result, we can conclude
that a minimal-length counterexample path is obtained when the actual
argument $x$ is $0$.

\section{Horn Constraint Optimization and Reduction from Refinement Type Optimization}
\label{sec:opt}

We reduce refinement type optimization problems into constraint
optimization problems subject to existentially-quantified Horn
clauses~\cite{Beyene2013,Unno2013,Kuwahara2015}.  We first formalize
Horn constraint optimization problems (in Section~\ref{sec:horn}) and
then explain the reduction (in Section~\ref{sec:red}).

\subsection{Horn Constraint Optimization Problems}
\label{sec:horn}

\emph{Existentially-Quantified Horn Clause Constraint Sets}
($\exists$HCCSs) over QFLIA are defined as follows.
\begin{align*}
\begin{array}{rrl}
(\text{$\exists$HCCSs}) & \hccs &::= \{\hc_1,\ldots,\hc_m\}\\
(\text{Horn clauses}) & \hc &::= h\Leftarrow b\\
(\text{heads}) & h &::= P(\seq{t}) \mid \phi \mid \exists \seq{x}.(P(\seq{t}) \land \phi) \\
(\text{bodies}) & b &::= P_{1}(\seq{t}_1)\land\ldots\land
 P_{m}(\seq{t}_m)\land \phi
\end{array}
\end{align*}
We write $\pvs(\hccs)$ for the set of predicate variables that occur
in $\hccs$.

A \emph{predicate substitution} $\subst$ for an $\exists$HCCS $\hccs$
is a map from each $P \in \pvs(\hccs)$ to a closed predicate $\lambda
x_1,\dots,x_{\arity{P}}.\phi$.  We write $\Theta_{\hccs}$ for the set
of predicate substitutions for $\hccs$.
We call a substitution $\subst$ is a \emph{solution} of $\hccs$ if
for each $\hc \in \hccs$, $\models \subst \hc$.
For a subset $\Theta \subseteq \Theta_{\hccs}$, we call a substitution
$\theta \in \Theta$ is a \emph{$\Theta$-restricted solution} if
$\theta$ is a solution of $\hccs$.  Our constraint optimization method
described in Section~\ref{sec:alg} is designed to find a
$\Theta$-restricted solution for some $\Theta$ consisting of
substitutions that map each predicate variable to a predicate with a
bounded number of conjunctions and disjunctions.  In particular, we
often use
\[
\SubstAtom=\set{P \mapsto \lambda x_1,\dots,x_{\arity{P}}. n_0 +
  \Sigma_{i=1}^{\arity{P}} n_i \cdot x_i \geq 0 \mid P \in
  \pvs(\hccs)}
\]
consisting of atomic predicate substitutions.

\begin{example}
\label{ex:sum_hco}
Recall the function $\SUM$ and the predicate substitutions $\subst_1,
\subst_2, \subst_3$ in Example~\ref{ex:sum}.
Our method reduces a type optimization problem for $\SUM$ into a
constraint optimization problem for the following HCCS $\hccs_{\SUM}$
(the explanation of the reduction is deferred to
Section~\ref{sec:red}).
\[
\set{
\begin{array}{l}
Q(x,0) \Leftarrow P(x) \land x = 0,\ P(x-1) \Leftarrow P(x) \land x \neq 0, \\
Q(x,x+y) \Leftarrow P(x) \land Q(x-1,y) \land x \neq 0
\end{array}}
\]

Here, $\subst_1$ is a solution of $\hccs_{\SUM}$, and $\subst_2$ and
$\subst_3$ are $\SubstAtom$-restricted solutions of $\hccs_{\SUM}$. If
we fix $Q(x,y) \equiv \bot$ (i.e., infer $\SUM$'s type of the form
$(x:\rt{x}{P(x)}) \to \rt{y}{\bot}$) for disproving termination of
$\SUM$ as in Section~\ref{sec:non-term}, we obtain the following HCCS
$\hccs_{\SUM}^{\bot}$.
\[
\set{\bot \Leftarrow P(x) \land x = 0,\ P(x-1)\Leftarrow P(x)\land x\neq 0}
\]
$\hccs_{\SUM}^{\bot}$ has, for example, $\SubstAtom$-restricted
solutions $\set{P\mapsto \lambda x. x<0}$ and $\{P\mapsto \lambda x. x
< -100\}$.
\qedex
\end{example}

We now define Horn constraint optimization problems for
$\exists$HCCSs.
\begin{definition}
\label{def:horn_opt}
Let $\hccs$ be an $\exists$HCCS and $\prec$ be a strict partial order
on predicate substitutions.  A solution $\subst$ of $\hccs$ is called
\emph{Pareto optimal with respect to $\prec$} if there is no solution
$\subst'$ of $\hccs$ such that $\subst' \prec \subst$.  A \emph{Horn
  constraint optimization problem} $(\hccs,\prec)$ is a problem to
find a Pareto optimal solution $\subst$ with respect to $\prec$.  A
\emph{$\Theta$-restricted Horn constraint optimization problem} is a
Horn constraint optimization problem with the notion of solutions
replaced by $\Theta$-restricted solutions.
\end{definition}

\begin{example}
Recall $\hccs_{\SUM}$ and its solutions
$\theta_1$,$\theta_2$,$\theta_3$ in Example~\ref{ex:sum}.  Let us
consider a Horn constraint optimization problem $(\hccs_{\SUM},
\prec_{(\rho,\sqsubset)})$ where $\rho(P)=\ \MAX$, $\rho(Q)=\ \MIN$,
and $Q \sqsubset P$.  We have $\theta_3 \prec_{(\rho,\sqsubset)}
\theta_1$ and $\theta_3 \prec_{(\rho,\sqsubset)} \theta_2$.  In fact,
$\theta_3$ is a Pareto optimal solution of $\hccs_{\SUM}$ with respect
to $\prec_{(\rho,\sqsubset)}$.  \qedex
\end{example}

In general, an $\exists$HCCS $\hccs$ may not have a Pareto optimal
solution with respect to $\prec_{(\rho,\sqsubset)}$ even though
$\hccs$ has a solution.  For example, consider a Horn constraint
optimization problem $(\hccs_{\SUM},\prec_{(\rho,\sqsubset)})$ where
$\rho(P) = \MAX$, $\rho(Q) = \MIN$, and $P\sqsubset Q$.  Because the
semantically optimal solution $Q(x,y) \equiv y = \frac{x(x+1)}{2}$ is
not expressible in QFLIA, it must be approximated, for example, as
$Q(x,y) \equiv y \geq 0 \land y \geq x \land y \geq 2x-1$.  The
approximated solution, however, is not Pareto optimal because we can
always get a better approximation like $Q(x,y) \equiv y \geq 0 \land y
\geq x \land y \geq 2x-1 \land y \geq 3x-3$ if we use more
conjunctions.

We can, however, show that an $\exists$HCCS has a
$\SubstAtom$-restricted Pareto optimal solution with respect to
$\prec_{(\rho,\sqsubset)}$ if it has a $\SubstAtom$-restricted
solution.  For the above example, $\subst_2$ in Example~\ref{ex:sum}
is a $\SubstAtom$-restricted Pareto optimal solution.
\begin{lemma}
\label{lem:exist}
Suppose that an $\exists$HCCS $\hccs$ has a $\SubstAtom$-restricted
solution and for any $P$ such that $\rho(P)=\MIN$, $P$ is not
existentially quantified in $\hccs$.  It then follows that $\hccs$ has
a $\SubstAtom$-restricted Pareto optimal solution with respect to
$\prec_{(\rho,\sqsubset)}$.
\end{lemma}
\begin{proofsketch}
We prove the lemma by contradiction.  Suppose that $\hccs$ has a
$\SubstAtom$-restricted solution but no Pareto optimal one.  It then
follows that there exist an infinite descending chain $\subst_1
\succ_{(\rho,\sqsubset)} \subst_2 \succ_{(\rho,\sqsubset)} \dots$ of
$\SubstAtom$-restricted solutions and a predicate variable $P$ such
that
\begin{itemize}
\item $\forall i \geq 1.\ \subst_i(P) \succ_{\rho(P)} \subst_{i+1}(P)
  \land \forall Q \sqsubset P.\ \subst_i(Q) \equiv_{\rho(Q)}
  \subst_{i+1}(Q)$ and
\item no $\SubstAtom$-restricted solution is a lower bound of the
  chain.
\end{itemize}
The key observations here are that the half-spaces represented by
$\subst_1(P),\subst_2(P),\dots$ are parallel, and for some $k>0$ that
depends on $2^{\arity{P}}$ and the largest absolute value of
coefficients in $\subst_1(P)$, the distance $d_i$ between
$\subst_i(P)$ and $\subst_{i+k}(P)$ are $d_i > 1$ for all $i\geq 1$,
because of the strictness of $\succ_{\rho(P)}$ and the discreteness of
integers.  By continuity of $\hccs$, $\hccs$ has a
$\SubstAtom$-restricted solution $\subst$ such that $\subst(P) \equiv
\lambda\seq{x}.\top$ if $\rho(P)=\MAX$ and $\subst(P) \equiv
\lambda\seq{x}.\bot$ if $\rho(P)=\MIN$, and $\forall Q \sqsubset
P.\ \subst(Q) \equiv_{\rho(Q)} \subst_1(Q)$.  $\subst$ is obviously a
lower bound of the chain.  Thus, a contradiction is obtained.  \qed
\end{proofsketch}

\subsection{Reduction from Refinement Type Optimization}
\label{sec:red}

Our method reduces a refinement type optimization problem into an Horn
constraint optimization problem in a similar manner to the previous
refinement type inference method~\cite{Unno2009}.  Given a program
$D$, our method first prepares a refinement type template $\Gamma_D$
of $D$ as well as, for each expression of the form $\LET\ x=
*_{\exists}\ \IN\ e$, a refinement type template
$\rt{x}{P(\seq{y},x)}$ of $x$, where $P$ is a fresh predicate variable
and $\seq{y}$ is the sequence of all integer variables in the scope.
Our method then generates an $\exists$HCCS by type-checking $D$
against $\Gamma_D$ and collecting the proof obligations of the forms
$\sem{\Gamma} \land \phi_1 \Implies \phi_2$ and $\sem{\Gamma} \Implies
\exists \nu.\phi$ respectively from each application of the rules
\rn{ISub} and \rn{Rand$\exists$}.  We write $\GEN(D,\Gamma_D)$ to
denote the $\exists$HCCS thus generated from $D$ and $\Gamma_D$.

We can show the soundness of our reduction in the same way as in
\cite{Unno2009}.
\begin{theorem}[Soundness of Reduction]
\label{thm:red}
Let $(D,\prec)$ be a refinement type optimization problem and
$\Gamma_D$ be a refinement type template of $D$.  If $\subst$ is a
Pareto optimal solution of $\GEN(D,\Gamma_D)$, then $\subst$ is a
solution of $(D,\prec)$.
\end{theorem}

\section{Horn Constraint Optimization Method}
\label{sec:alg}

\begin{figure}[t]
 \centering
 \begin{minipage}[c]{0.76\textwidth}
\begin{algorithmic}[1]
\Procedure{Optimize}{$\hccs,\prec$}
 \State \Match{\textsc{Solve}$(\hccs)$}\label{pse:solve}
 \State $\ \;\;$ \textit{Unknown} $\to$ \ret \textit{Unknown}\label{pse:unknown}
 \State $\ \mid\;$ \textit{NoSol} $\to$ \ret \textit{NoSol}\label{pse:nosol}
 \State $\ \mid\;$ \textit{Sol}$(\subst_0) \to$\label{pse:sol} \\
 \ind\ind $\subst := \subst_0$;\label{pse:improve_begin} \\
 \ind\ind \textbf{while} $\TRUE$ \textbf{do} \label{pse:while_begin}\\
 \ind\ind\ind
 \Let{$\hccs'$}{$\textsc{Improve}_{\prec}(\subst,\hccs)$}\label{pse:generate} \\
 \ind\ind\ind \Match{\textsc{Solve}$(\hccs')$}
 \State\ind\ind $\ \;\;$ \textit{Unknown} $\to$ \ret \textit{Sol}$(\subst)$\label{pse:improve_unknown}
 \State\ind\ind $\ \mid\;$ \textit{NoSol} $\to$ \ret \textit{OptSol}$(\subst)$\label{pse:improve_nosol}
 \State\ind\ind $\ \mid\;$ \textit{Sol}$(\subst') \to \subst := \subst'$\label{pse:improve_sol} \\
 \ind\ind \textbf{end}\label{pse:improve_end}
 \EndProcedure
\end{algorithmic}
 \end{minipage}
 \caption{Pseudo-code of the constraint optimization method for $\exists$HCCSs}
 \label{fig:code_optimize}
\end{figure}

In this section, we describe our Horn constraint optimization method
for $\exists$HCCSs.  The method repeatedly improves a current solution
until convergence.  The pseudo-code of the method is shown in
Figure~\ref{fig:code_optimize}.  The procedure \textsc{Optimize} for
Horn constraint optimization takes a ($\Theta$-restricted)
$\exists$HCCS optimization problem $(\hccs,\prec)$ and returns any of
the following: \textit{Unknown} (which means the existence of a
solution is unknown), \textit{NoSol} (which means no solution exists),
\textit{Sol}$(\theta)$ (which means $\theta$ is a possibly non-Pareto
optimal solution), or \textit{OptSol}$(\theta)$ (which means $\theta$
is a Pareto optimal solution).  The sub-procedure \textsc{Solve} for
Horn constraint solving takes an $\exists$HCCS $\hccs$ and returns any
of \textit{Unknown}, \textit{NoSol}, or \textit{Sol}$(\theta)$.  The
detailed description of \textsc{Solve} is deferred to
Section~\ref{sub:solve}.

\textsc{Optimize} first calls \textsc{Solve} to find an initial
solution $\theta_0$ of $\hccs$ (line \ref{pse:solve}).
\textsc{Optimize} returns $\textit{Unknown}$ if \textsc{Solve} returns
$\textit{Unknown}$ (line \ref{pse:unknown}) and $\textit{NoSol}$ if
\textsc{Solve} returns $\textit{NoSol}$ (line \ref{pse:nosol}).
Otherwise (line \ref{pse:sol}), \textsc{Optimize} repeatedly improves
a current solution $\theta$ starting from $\theta_0$ until convergence
(lines \ref{pse:improve_begin} -- \ref{pse:improve_end}).
To improve $\theta$, we call a sub-procedure
$\textsc{Improve}_{\prec}(\theta,\hccs)$ for generating an
$\exists$HCCS $\hccs'$ from $\hccs$ by adding constraints that require
any solution $\theta'$ of $\hccs'$ satisfies $\theta' \prec \theta$
(line \ref{pse:generate}).  \textsc{Optimize} then calls
\textsc{Solve} to find a solution of $\hccs'$.  If \textsc{Solve}
returns $\textit{Unknown}$, \textsc{Optimize} returns
\textit{Sol}$(\theta)$ as a (possibly non-Pareto optimal) solution
(line \ref{pse:improve_unknown}).  If \textsc{Solve} returns
$\textit{NoSol}$, it is the case that no improvement is possible, and
hence the current solution $\theta$ is Pareto optimal.  Thus,
\textsc{Optimize} returns \textit{OptSol}$(\theta)$ (line
\ref{pse:improve_nosol}).  Otherwise, we obtain an improved solution
$\theta' \prec \theta$, and \textsc{Optimize} updates the current
solution $\theta$ with $\theta'$ and repeats the improvement process
(line \ref{pse:improve_sol}).

\begin{example}
\label{ex:optimize}
Recall $\hccs_{\SUM}^{\bot}$ in Example~\ref{ex:sum_hco} and consider
an optimization problem
$(\hccs_{\SUM}^{\bot},\prec_{(\sqsubset,\rho)})$ where $\rho(P)=\MAX$.
We below explain how
$\textsc{Optimize}(\hccs_{\SUM}^{\bot},\prec_{(\sqsubset,\rho)})$
proceeds.  First, $\textsc{Optimize}$ calls the sub-procedure
$\textsc{Solve}$ to find an initial solution of $\hccs_{\SUM}^{\bot}$
(e.g., $\theta_0 = \set{P\mapsto \lambda x.\bot}$).
$\textsc{Optimize}$ then calls the sub-procedure
$\textsc{Improve}_{\prec}(\theta_0,\hccs_{\SUM}^{\bot})$ and obtains
an $\exists$HCCS $\hccs'=\hccs_{\SUM}^{\bot} \cup \{P(x)\Leftarrow
\bot,\; \exists x.P(x) \land \neg \bot\}$.  Note that $\hccs'$
requires that for any solution $\theta$ of $\hccs'$, $\theta(P)
\prec_{\rho(P)} \theta_0(P) = \lambda x.\bot$.  $\textsc{Optimize}$
then calls $\textsc{Solve}(\hccs')$ to find an improved solution of
$\hccs$ (e.g., $\theta_1 = \set{ P\mapsto\lambda x. x < 0 }$).  In the
next iteration, $\textsc{Optimize}$ returns $\theta_1$ as a Pareto
optimal solution because
$\textsc{Improve}_{\prec}(\theta_1,\hccs_{\SUM}^{\bot})$ has no
solution.  \qedex
\end{example}

We now discuss properties of the procedure \textsc{Optimize} under the
assumption of the correctness of the sub-procedure \textsc{Solve}
(i.e., $\subst$ is a $\Theta$-restricted solution of $\hccs$ if
$\textsc{Solve}(\hccs)$ returns \textit{Sol}$(\subst)$, and $\hccs$
has no $\Theta$-restricted solution if \textsc{Solve}$(\hccs)$ returns
\textit{NoSol}).  The following theorem states the correctness of
\textsc{Optimize}.
\begin{theorem}[Correctness of the Procedure \textsc{Optimize}]
\label{thm:opt}
Let $(\hccs,\prec)$ be a $\Theta$-restricted Horn constraint
optimization problem.  If $\textsc{Optimize}(\hccs,\prec)$ returns
\textit{OptSol}$(\theta)$ (resp. \textit{Sol}$(\theta)$), $\theta$ is
a Pareto optimal (resp. possibly non-Pareto optimal)
$\Theta$-restricted solution of $\hccs$ with respect to $\prec$.
\end{theorem}

The following theorem states the termination of $\textsc{Optimize}$
for $\SubstAtom$-restricted Horn constraint optimization problems.
\begin{theorem}[Termination of the Procedure \textsc{Optimize}]
\label{thm:term}
Let $(\hccs,\prec_{(\sqsubset,\rho)})$ be a $\SubstAtom$-restricted
Horn constraint optimization problem.  It then follows that
\textsc{Optimize}$(\hccs,\prec_{(\sqsubset,\rho)})$ always terminates
if the sub-procedure \textsc{Solve} preferentially returns solutions
having smaller absolute values of coefficients.
\end{theorem}
\begin{proof}
Recall the proof sketch of Lemma~\ref{lem:exist}.  Any infinite
descending chain of $\SubstAtom$-restricted solutions for a predicate
variable $P$ with respect to $\prec_{(\sqsubset,\rho)}$ has a limit
$\lambda\seq{x}.\top$ if $\rho(P)=\MAX$ and $\lambda\seq{x}.\bot$ if
$\rho(P)=\MIN$.  Because $\lambda\seq{x}.\bot$
(resp. $\lambda\seq{x}.\top$) is expressed as an atomic predicate
$\lambda\seq{x}.-1 \geq 0$ (resp. $\lambda\seq{x}.0 \geq 0$) having
absolute values of coefficients not greater than $1$ and the number of
such predicates is finite, the limit is guaranteed to be reached in a
finite number of iterations.  \qed
\end{proof}

\subsection{Sub-Procedure \textsc{Solve} for Solving $\exists$HCCSs}
\label{sub:solve}

The pseudo-code of the sub-procedure \textsc{Solve} for solving
$\exists$HCCSs is presented in Figure~\ref{fig:code_solve}.  Here,
\textsc{Solve} uses existing template-based invariant generation
techniques based on Farkas' lemma~\cite{Colon2003,Gulwani2008a} and
$\exists$HCCS solving techniques based on
Skolemization~\cite{Unno2013,Beyene2013,Kuwahara2015}.  \textsc{Solve}
first generates a template substitution $\theta$ that maps each
predicate variable in $\pvs(\hccs)$ to a template atomic predicate
with unknown coefficients $c_0,\dots,c_{\arity{P}}$ (line
\ref{pse:init}).\footnote{In this way, the particular code is
  specialized to solve $\SubstAtom$-restricted Horn constraint
  optimization problems.  To solve $\Theta$-restricted optimization
  problems for other $\Theta$, we need here to generate templates that
  conform to the shape of substitutions in $\Theta$ instead.  Our
  implementation in Section~\ref{sec:exp} iteratively increases the
  template size.}  \textsc{Solve} then applies $\theta$ to $\hccs$ and
obtains a verification condition of the form $\exists \seq{c}.\forall
\seq{x}.\exists \seq{y}.\phi$ without predicate variables (line
\ref{pse:subst}).  \textsc{Solve} applies
Skolemization~\cite{Unno2013,Beyene2013,Kuwahara2015} to the condition
and obtains a simplified condition of the form $\exists
\seq{c},\seq{z}.\forall \seq{x}.\phi'$ (line \ref{pse:skolem}).
\textsc{Solve} further applies Farkas'
lemma~\cite{Colon2003,Gulwani2008a} to eliminate the universal
quantifiers and obtains a condition of the form $\exists
\seq{c},\seq{z},\seq{w}.\phi''$ (line \ref{pse:farkas}).
\textsc{Solve} then uses an off-the-shelf SMT solver to find a
satisfying assignment to $\phi''$ (line \ref{pse:smt}).  If such an
assignment $\sigma$ is found, \textsc{Solve} returns $\sigma(\theta)$
as a solution (line \ref{pse:solve_sat}).  \textsc{Solve} returns
\textit{Unknown} if the SMT solver returns \textit{Unknown} (line
\ref{pse:solve_unknown}).  Otherwise (no assignment is
found),\footnote{Note here that even though no assignment is found,
  $\hccs$ may have a $\SubstAtom$-restricted solution because Farkas'
  lemma is not complete for QFLIA
  formulas~\cite{Colon2003,Gulwani2008a} and Skolemization of $\exists
  \seq{c}.\forall \seq{x}.\exists \seq{y}.\phi$ into $\exists
  \seq{c},\seq{z}.\forall \seq{x}.\phi'$ here assumes that $\seq{y}$
  are expressed as linear expressions over
  $\seq{x}$~\cite{Unno2013,Beyene2013,Kuwahara2015}.}  \textsc{Solve}
uses the SMT solver again to find a satisfying assignment $\sigma$ to
$\forall\seq{x}.\exists \seq{y}.\phi$ (line \ref{pse:solve_unsat}).
If such a $\sigma$ is found, \textsc{Solve} returns $\sigma(\theta)$
as a solution (line \ref{pse:solve_unsat_sat}).  \textsc{Solve}
returns \textit{Unknown} if \textit{Unknown} is returned (line
\ref{pse:solve_unsat_unknown}) and \textit{NoSol} if \textit{Unsat} is
returned (line \ref{pse:solve_unsat_unsat}).

\begin{figure}[t]
 \centering
 \begin{minipage}[c]{0.76\textwidth}
 \begin{algorithmic}[1]
 \Procedure{Solve}{$\hccs$}
 \State \Let{$\subst$}{$\set{P \mapsto \lambda \seq{x}.
                                       c_0 + \Sigma_{i=1}^{\arity{P}} c_i \cdot x_i \geq 0 \mid P \in \pvs(\hccs)}$}\label{pse:init}
 \State \Let{$\exists \seq{c}.\forall \seq{x}.\exists \seq{y}.\phi$}{$\exists \seq{c}.\bigwedge_{\hc \in \hccs} \forall \fvs(\hc).\subst(\hc)$}\label{pse:subst}
 \State \Let{$\exists \seq{c},\seq{z}.\forall \seq{x}.\phi'$}{apply Skolemization to $\exists \seq{c}.\forall \seq{x}.\exists \seq{y}.\phi$}\label{pse:skolem}
 \State \Let{$\exists \seq{c},\seq{z},\seq{w}.\phi''$}{apply Farkas' lemma to $\exists \seq{c},\seq{z}.\forall \seq{x}.\phi'$}\label{pse:farkas}
 \State \Match{\textsc{SMT}($\phi''$)}\label{pse:smt}
 \State $\ \;\;$ \textit{Unknown} $\to$ \ret \textit{Unknown}\label{pse:solve_unknown}
 \State $\ \mid\;$ \textit{Sat}$(\sigma) \to$ \ret \textit{Sol}$(\sigma(\theta))$\label{pse:solve_sat}
 \State $\ \mid\;$ \textit{Unsat} $\to$ \Match{\textsc{SMT}($\forall
  \seq{x}.\exists \seq{y}.\phi$)}\label{pse:solve_unsat}
  \State $\qquad\qquad\qquad\ $ \textit{Unknown} $\to$ \ret \textit{Unknown}\label{pse:solve_unsat_unknown}
  \State $\qquad\qquad\quad\ \ \mid\;\,$ \textit{Sat}$(\sigma) \to$ \ret \textit{Sol}$(\sigma(\theta))$\label{pse:solve_unsat_sat}
  \State $\qquad\qquad\quad\ \ \mid\;\,$ \textit{Unsat} $\to$ \ret \textit{NoSol}\label{pse:solve_unsat_unsat}
 \EndProcedure
 \end{algorithmic}
 \end{minipage}
 \caption{Pseudo-code of the constraint solving method for $\exists$HCCSs based on template-based invariant generation}
 \label{fig:code_solve}
\end{figure}

\begin{example}
We explain how $\textsc{Solve}$ proceeds for $\hccs'$ in
Example~\ref{ex:optimize}.  $\textsc{Solve}$ first generates a
template substitution $\theta = \set{ P\mapsto\lambda x. c_0 +
  c_1\cdot x \geq 0}$ with unknown coefficients $c_0,c_1$ and applies
$\theta$ to $\hccs'$.  As a result, we get a verification condition
\[
\exists c_0,c_1.
\left(
\begin{array}{l}
\forall x.
 \left(
 \begin{array}{l}
 (\bot \Leftarrow c_0 + c_1\cdot x \geq 0\, \land\, x=0)\,\land\,\\
 (c_0 + c_1\cdot(x-1) \geq 0\Leftarrow c_0 + c_1\cdot x \geq 0\,\land\, x \neq 0)
 \end{array}
 \right)\,\land\,\\
\exists x.\ c_0 + c_1\cdot x \geq 0
\end{array}
\right)
\]
By applying Farkas' lemma, we obtain
\[
\exists c_0,c_1.\left(
\begin{array}{l}
\exists w_1, w_2, w_3 \geq 0.\ (c_0\cdot w_1 \leq -1 \land c_1\cdot w_1 + w_2 - w_3 = 0)\,\land\,\\
\exists w_4, w_5, w_6 \geq 0.
  \left(
  \begin{array}{l}
  (-1-c_0+c_1)\cdot w_4 + c_0\cdot w_5 - w_6 \leq -1\,\land\\
  c_1\cdot (-w_4 + w_5) + w_6 = 0
  \end{array}
  \right)\,\land\,\\
\exists w_7, w_8, w_9 \geq 0.
  \left(
  \begin{array}{l}
  (-1-c_0+c_1)\cdot w_7 + c_0\cdot w_8 - w_9 \leq -1\,\land\\
  c_1\cdot (-w_7 + w_8) - w_9 = 0
\end{array}
\right)\,\land\,\\
\exists x.\ c_0 + c_1\cdot x \geq 0
\end{array}
\right)
\]
By using an SMT solver, we obtain, for example, a satisfying
assignment
\[
\sigma=\left\{
\begin{array}{l}
c_0 \mapsto -1, c_1 \mapsto -1,
w_1 \mapsto 1, w_2 \mapsto 1, w_3 \mapsto 0,\\
w_4 \mapsto 0, w_5 \mapsto 1, w_6 \mapsto 1,
w_7 \mapsto 1, w_8 \mapsto 0, w_9 \mapsto 1
\end{array}
\right\}
\]
Thus, \textsc{Solve} returns $\sigma(\theta)=\set{P\mapsto\lambda
  x. -1-x \geq 0} \equiv \theta_1$ in Example~\ref{ex:optimize}.
\qedex
\end{example}

The following theorem states the correctness of the sub-procedure
\textsc{Solve}.
\begin{lemma}[Correctness of the Sub-Procedure \textsc{Solve}]
\label{lem:solve}
Let $\hccs$ be an $\exists$HCCS.  $\subst$ is a
$\SubstAtom$-restricted solution of $\hccs$ if \textsc{Solve}$(\hccs)$
returns \textit{Sol}$(\subst)$, and $\hccs$ has no
$\SubstAtom$-restricted solution if \textsc{Solve}$(\hccs)$ returns
\textit{NoSol}.
\end{lemma}

Note that the sub-procedure \textsc{Solve} described above does not
necessarily satisfy the assumption of Theorem~\ref{thm:term}.  We can,
however, extend \textsc{Solve} to satisfy the assumption by bounding
the absolute values of unknown coefficients and iteratively
incrementing the bounds in SMT solving.

\section{Implementation and Experiments}
\label{sec:exp}

We have implemented a prototype refinement type optimization tool
based on our method.  Our tool takes OCaml programs and uses
Z3~\cite{Moura2008} as the underlying SMT solver in the sub-procedure
\textsc{Solve}.
We conducted preliminary experiments for each application presented in
Section~\ref{sec:app}.  All the experiments were conducted on a
machine with Intel Core i7-4650U 1.70GHz, 8GB of RAM.

The experimental results are summarized in
Tables~\ref{tbl:result_known} and \ref{tbl:result_all}.
Table~\ref{tbl:result_known} shows the results of an existing
first-order non-termination verification benchmark set used in
\cite{Larraz2014,Chen2014,Kuwahara2015}.  Because the original
benchmark set was written in the input language of T2
(\url{http://mmjb.github.io/T2/}), we used an OCaml translation of the
benchmark set provided by \cite{Kuwahara2015}.
Our tool was able to successfully disprove termination of 41 programs
(out of 81) in the time limit of 100 seconds.  Our prototype tool was
not the best but performed well compared to the state-of-the-art tools
dedicated to non-termination verification.

\begin{table}[t]
 \caption{The results of a non-termination verification benchmark set
   used in \cite{Larraz2014,Chen2014,Kuwahara2015}.  The results for
   \textsc{CppInv}, T2-TACAS, and TNT are according to Larraz et
   al.~\cite{Larraz2014}.  The result for MoCHi is according to
   \cite{Kuwahara2015}.}  \centering
 \begin{tabular}{|c||c|c|c|} \hline
  & Verified & TimeOut & Other \\ \hline
  Our tool & 41 & 27 & 13 \\ \hline
  \textsc{CppInv}~\cite{Larraz2014} & 70 & 6 & 5 \\ \hline
  T2-TACAS~\cite{Chen2014} & 51 & 0 & 30 \\ \hline
  TNT~\cite{Emmes2012} & 19 & 3 & 59\\ \hline
  MoCHi~\cite{Kuwahara2015} & 48 & 26 & 7 \\ \hline
 \end{tabular}
 \label{tbl:result_known}
\end{table}

\newcommand{\dist}{Disprove Termination}
\newcommand{\diss}{Disprove Safety}
\newcommand{\safe}{Prove Safety}
\newcommand{\term}{Prove Termination}

\newcommand{\ho}{H}
\newcommand{\fo}{F}

\begin{table}[t]
 \caption{The results of maximally-weak precondition inference for
   proving/disproving safety/termination.}
 \centering
 \begin{tabular}{|r|c|r|r|} \hline
  Program & Applications & \#Iterations & Time (ms) \\\hline
  \verb|fixpoint_nonterm|\cite{Kuwahara2015} & \dist & 1 & 2,020\\
  \verb|fib_CPS_nonterm|\cite{Kuwahara2015} & \dist & 1 & 5,023\\
  \verb|indirect_e|\cite{Kuwahara2015} & \dist & 1 & 1,083 \\
  \verb|indirectHO_e|\cite{Kuwahara2015} & \dist & 1 & 2,434 \\
  \verb|loopHO|\cite{Kuwahara2015} & \dist & 1 & 1,642 \\
  \verb|foldr_nonterm|\cite{Kuwahara2015} & \dist & 1 & 4,904 \\
\hline
  \verb|repeat| (Sec.~\ref{sec:safe}) &\safe & 4 & 948 \\
  \verb|sum_geq3| & \safe & 4 & 2,654 \\
  \verb|append| & \safe & 5 & 20,352 \\
\hline
  \verb|append| \cite{Kuwahara2014} & \term & 6 & 26,786 \\
  \verb|zip| \cite{Kuwahara2014} & \term & 13 & 76,641 \\
\hline
  $\SUMP$ (Sec.~\ref{sec:safe}) & \safe & 3 & 1,856 \\
  $\SUM$ (Sec.~\ref{sec:non-term}) & \dist & 1 & 174 \\
  $\SUMT$ (Sec.~\ref{sec:term}) & \term ($P \sqsubset \mathit{Inv}$) & 4 & 56,042 \\
  $\SUMT$ (Sec.~\ref{sec:term}) & \term ($\mathit{Inv} \sqsubset P$) & 3 & 6,628 \\
  $\SUMT$ (Sec.~\ref{sec:non-safe}) & \diss ($P \sqsubset \mathit{Inv}$) & $>2$ & 18,009\\
  $\SUMT$ (Sec.~\ref{sec:non-safe}) & \diss ($\mathit{Inv} \sqsubset P$) & $>2$ & 16,540 \\
  \hline
 \end{tabular}
 \label{tbl:result_all}
\end{table}

Table~\ref{tbl:result_all} shows the results of maximally-weak
precondition inference for proving safety and termination, and
disproving safety and termination.  In the column \#Iterations, $>2$
represents that the 3-rd iteration timed out and possibly non-Pareto
optimal solution was inferred by our tool.  We used non-termination
(resp. termination) verification benchmarks for higher-order programs
from \cite{Kuwahara2015} (resp. \cite{Kuwahara2014}).  The results
show that our method is also effective for safety and non-termination
verification of higher-order programs.  Our prototype tool, however,
could be optimized further to speed up termination and non-safety
verification.

\section{Related Work}
\label{sec:rel}

Type inference problems of refinement type
systems~\cite{Freeman1991,Xi1999} have been intensively
studied~\cite{Unno2008,Rondon2008,Unno2009,Kobayashi2011b,Jhala2011,Terauchi2010,Unno2013}.
To our knowledge, this paper is the first to address type optimization
problems, which generalize ordinary type inference problems.  As we
saw in Sections~\ref{sec:app} and \ref{sec:exp}, this generalization
enables significantly wider applications in the verification of
higher-order functional programs.

For imperative programs, Gulwani et al. have proposed a template-based
method to infer maximally-weak pre and maximally-strong post
conditions~\cite{Gulwani2008a}.  Their method, however, cannot
directly handle higher-order functional programs, (angelic and
demonic) non-determinism in programs, and prioritized multi-objective
optimization, which are all handled by our new method.

Internally, our method reduces a type optimization problem to a
constraint optimization problem subject to an existentially quantified
Horn clause constraint set ($\exists$HCCS).
Constraint \emph{solving} problems for $\exists$HCCSs have been
studied by recent work~\cite{Unno2013,Beyene2013,Kuwahara2015}.  They,
however, do not address constraint \emph{optimization} problems.
The goal of our constraint optimization is to maximize/minimize the
set of the models for each predicate variable occurring in the given
$\exists$HCCS.  Thus, our constraint optimization problems are
different from Max-SMT~\cite{Nieuwenhuis2006} problems whose goal is
to minimize the sum of the penalty of unsatisfied clauses.

\section{Conclusion}
\label{sec:concl}

We have generalized refinement type inference problems to type
optimization problems, and presented interesting applications enabled
by type optimization to inferring most-general characterization of
inputs for which a given functional program satisfies (or violates) a
given safety (or termination) property.  We have also proposed a
refinement type optimization method based on template-based invariant
generation.  We have implemented our method and confirmed by
experiments that the proposed method is promising for the
applications.

\bibliography{abbrv,prog_lang,bib}

\begin{thebibliography}{10}

\bibitem{Beyene2013}
T.~A. Beyene, C.~Popeea, and A.~Rybalchenko.
\newblock Solving existentially quantified horn clauses.
\newblock In {\em CAV '13}, volume 8044 of {\em LNCS}, pages 869--882.
  Springer, 2013.

\bibitem{Chen2014}
H.~Y. Chen, B.~Cook, C.~Fuhs, K.~Nimkar, and P.~W. O'Hearn.
\newblock Proving nontermination via safety.
\newblock In {\em TACAS '14}, volume 8413 of {\em LNCS}, pages 156--171.
  Springer, 2014.

\bibitem{Colon2003}
M.~A. Col\'{o}n, S.~Sankaranarayanan, and H.~B. Sipma.
\newblock Linear invariant generation using non-linear constraint solving.
\newblock In {\em CAV '03}, volume 2725 of {\em LNCS}, pages 420--432.
  Springer, 2003.

\bibitem{Moura2008}
L.~de~Moura and N.~Bj{\o}rner.
\newblock Z3: An efficient {SMT} solver.
\newblock In {\em TACAS '08}, volume 4963 of {\em LNCS}, pages 337--340.
  Springer, 2008.

\bibitem{Emmes2012}
F.~Emmes, T.~Enger, and J.~Giesl.
\newblock Proving non-looping non-termination automatically.
\newblock In {\em IJCAR '12}, volume 7364 of {\em LNCS}, pages 225--240.
  Springer, 2012.

\bibitem{Freeman1991}
T.~Freeman and F.~Pfenning.
\newblock Refinement types for {ML}.
\newblock In {\em PLDI '91}, pages 268--277. ACM, 1991.

\bibitem{Gulwani2009a}
S.~Gulwani, K.~K. Mehra, and T.~Chilimbi.
\newblock {SPEED}: Precise and efficient static estimation of program
  computational complexity.
\newblock In {\em POPL '09}, pages 127--139. ACM, 2009.

\bibitem{Gulwani2008a}
S.~Gulwani, S.~Srivastava, and R.~Venkatesan.
\newblock Program analysis as constraint solving.
\newblock In {\em PLDI '08}, pages 281--292. ACM, 2008.

\bibitem{Jhala2011}
R.~Jhala, R.~Majumdar, and A.~Rybalchenko.
\newblock {HMC}: verifying functional programs using abstract interpreters.
\newblock In {\em CAV '11}, volume 6806 of {\em LNCS}, pages 470--485.
  Springer, 2011.

\bibitem{Kobayashi2011b}
N.~Kobayashi, R.~Sato, and H.~Unno.
\newblock Predicate abstraction and {CEGAR} for higher-order model checking.
\newblock In {\em PLDI '11}, pages 222--233. ACM, 2011.

\bibitem{Kuwahara2015}
T.~Kuwahara, R.~Sato, H.~Unno, and N.~Kobayashi.
\newblock Predicate abstraction and {CEGAR} for disproving termination of
  higher-order functional programs.
\newblock In {\em CAV'15}, LNCS. Springer, 2015.

\bibitem{Kuwahara2014}
T.~Kuwahara, T.~Terauchi, H.~Unno, and N.~Kobayashi.
\newblock Automatic termination verification for higher-order functional
  programs.
\newblock In {\em ESOP '14}, volume 8410 of {\em LNCS}, pages 392--411.
  Springer, 2014.

\bibitem{Larraz2014}
D.~Larraz, K.~Nimkar, A.~Oliveras, E.~Rodr\'{\i}guez-Carbonell, and A.~Rubio.
\newblock Proving non-termination using max-{SMT}.
\newblock In {\em CAV '14}, volume 8559 of {\em LNCS}, pages 779--796.
  Springer, 2014.

\bibitem{Nieuwenhuis2006}
R.~Nieuwenhuis and A.~Oliveras.
\newblock On {SAT} modulo theories and optimization problems.
\newblock In {\em SAT '06}, volume 4121 of {\em LNCS}, pages 156--169.
  Springer, 2006.

\bibitem{Rondon2008}
P.~Rondon, M.~Kawaguchi, and R.~Jhala.
\newblock Liquid types.
\newblock In {\em PLDI '08}, pages 159--169. ACM, 2008.

\bibitem{Terauchi2010}
T.~Terauchi.
\newblock Dependent types from counterexamples.
\newblock In {\em POPL '10}, pages 119--130. ACM, 2010.

\bibitem{Unno2008}
H.~Unno and N.~Kobayashi.
\newblock On-demand refinement of dependent types.
\newblock In {\em FLOPS '08}, volume 4989 of {\em LNCS}, pages 81--96.
  Springer, 2008.

\bibitem{Unno2009}
H.~Unno and N.~Kobayashi.
\newblock Dependent type inference with interpolants.
\newblock In {\em PPDP '09}, pages 277--288. ACM, 2009.

\bibitem{Unno2013}
H.~Unno, T.~Terauchi, and N.~Kobayashi.
\newblock Automating relatively complete verification of higher-order
  functional programs.
\newblock In {\em POPL '13}, pages 75--86. ACM, 2013.

\bibitem{Xi1999}
H.~Xi and F.~Pfenning.
\newblock Dependent types in practical programming.
\newblock In {\em POPL '99}, pages 214--227. ACM, 1999.

\end{thebibliography}
\bibliographystyle{abbrv}

\end{document}